\documentclass{amsart}
\usepackage{amsmath,amscd,amsthm,amsfonts}

\theoremstyle{plain}
\newtheorem{theorem}                 {Theorem}      [section]
\newtheorem{proposition}  [theorem]  {Proposition}

\newtheorem{lemma}        [theorem]  {Lemma}

\theoremstyle{definition}

\newtheorem{remark}       [theorem]  {Remark}
\newtheorem{definition}   [theorem]  {Definition}

\DeclareMathOperator{\divergence}{div}
\DeclareMathOperator{\const}{const}

\newcommand{\definedas}{\mathrel{\raise.095ex\hbox{\rm :}\mkern-5.2mu=}}

\let\oldmarginpar\marginpar
\renewcommand\marginpar[1]{\-\oldmarginpar[\raggedleft\tiny #1]%
{\raggedright\tiny #1}}

\begin{document}

\baselineskip 18pt \larger

\def \theo-intro#1#2 {\vskip .25cm\noindent{\bf Theorem #1\ }{\it #2}}

\newcommand{\trace}{\operatorname{tr}}

\def \rn{\mathbb R}
\def \hn{\mathbb H}

\def \R{\mathcal R}
\def \A{\mathcal A}
\def \B{\mathcal B}
\def \D{\mathcal D}

\allowdisplaybreaks
\title{Constant mean curvature solutions \\of the Einstein-scalar field constraint equations\\ on
  asymptotically hyperbolic manifolds }

\author{Anna Sakovich}


\subjclass[2000]{}

\address{Anna Sakovich, Institutionen f\"or Matematik, Kungliga Tekniska H\"ogskolan, 100 44
Stockholm, Sweden} \email{sakovich@math.kth.se}

\begin{abstract}
We follow the approach employed by  Y.~Choquet-Bruhat, J.~Isenberg and
D.~Pollack in the case of closed manifolds
and establish existence and non-existence results for constant mean curvature solutions of the
Einstein-scalar field constraint equations on asymptotically
hyperbolic manifolds.
\end{abstract}

\maketitle

\section{Introduction}

The constraint equations play an important role in the analysis of the
Einstein field equations of General Relativity. Once a set of initial
data which satisfies the
constraint equations is known, then by the fundamental theorem of Y. Choquet-Bruhat \cite{Ch-Br} and  its extension by Y. Choquet-Bruhat and R.~Geroch \cite{Ch-Br-Ge} there exists  a spacetime which solves
the Einstein equations. The constraint equations have been 	
thoroughly studied in the context of vacuum spacetimes (see
\cite{Ba-Is} for a comprehensive review), and recently a number of
results have appeared regarding the constraint equations for
Einstein-scalar field theories. In particular, we would like to mention the 
 works concerning respectively constant mean curvature (CMC)
solutions on closed manifolds \cite{Ch-Br-Is-Po} and \cite{He-Pa-Po}
and on asymptotically Euclidean manifolds \cite{Ch-Br-Is-Po-2}. 

In the light of these recent developments, an interesting task is to
analyze solvability of the Einstein-scalar field constraint equations
on asymptotically hyperbolic manifolds. Intuitively, these spaces can
be described as noncompact Riemannian manifolds with the metric
approaching a metric of constant negative curvature as one approaches
infinity.

It was conjectured in \cite{Ch-Br-Is-Po} that in the CMC case one can
effectively analyze the Einstein-scalar field constraint equations on
asymptotically hyperbolic manifolds using the same strategy as in the
case of closed manifolds. This approach is followed in the current
paper. Owing to conformal relatedness of asymptotically hyperbolic
manifolds to manifolds with negative scalar curvature \cite{An-Cr-Fr},
it makes sense to split the set of conformal data into subclasses
depending on the possible signs for the coefficients of its
terms. This splitting is used as the basis for proving Theorem
\ref{non-existence} and Theorem \ref{existence}, concerning 
non-existence and existence of CMC solutions of the Einstein-scalar
field constraint equations on asymptotically hyperbolic geometries respectively.

\section{The Einstein-scalar field constraint equations and asymptotically hyperbolic manifolds}

\subsection{The conformal method}

Consider an $(n+1)$-dimensional manifold $N$ with a spacetime metric
$\gamma$ and a real-valued scalar field $\Psi$.

Given an $n$-dimensional spacelike hypersurface $M$ in $N$, let 
$\bar h$ be its induced metric and $\bar K$ its second fundamental
form. Denote by $\bar \psi$ the restriction of the scalar field $\Psi$ 
to $M$ and by $\bar \pi$ the value of the derivative of $\Psi$ in the
direction of the unit normal of $M$ in $(N,\gamma)$.

Then the Einstein-scalar field constraint equations \cite{Ch-Br-Is-Po} comprise the
Hamiltonian constraint
\begin{equation}\label{constr1}
R_{\bar{h}}-\left|\bar K \right|^2 _{\bar h }
+ \left(\trace \bar K \right)^2
=
\bar{\pi}^2+\left|\nabla ^{\bar h} \bar \psi \right|^2 _{\bar h}+2V(\bar \psi),
\end{equation}
and the momentum constraint
\begin{equation}\label{constr2}
\divergence _{\bar h} \bar K - \nabla ^{\bar h} 
\left(\trace \bar K \right)
=
-\bar \pi \nabla ^{\bar h} \bar \psi,
\end{equation}
where all derivatives and norms are taken with respect to the metric
$\bar h$ on $M$, the potential $V$ 
is a smooth function of a real variable, and $R_{\bar h}$ denotes the
scalar curvature of $\bar h$. Note that in general it is not required that 
$V(0)$ should be equal to zero. For example, one can consider potentials with 
strictly positive minimum, which lead to accelerated expansion in cosmological 
models \cite{Re}.

If one can solve these equations for the Cauchy data $(\bar h, \bar K,
\bar \psi, \bar \pi)$ defined on a chosen $n$-dimensional manifold $M$
then there exists an $(n+1)$-dimensional spacetime solution 
$(M \times \rn, \gamma, \Psi)$ of the Einstein-scalar field equations which
is consistent with the given Cauchy data on $M$ (see, for example, \cite[Theorem 14.2]{Hans}). 

There is a standard procedure for rewriting the constraint equations
in a form which is more suitable for analysis, namely the conformal
method \cite{Ch-Br-Is-Po}. The idea is to split the Cauchy data on $M$
into (i) the (freely chosen) conformal background data, 
which in the scalar field
case consist of a Riemannian metric $h$, a symmetric trace-free and
divergence-free (0,2)-tensor $\sigma$ and scalar functions $\tau$,
$\psi$ and $\pi$ on $M$, and (ii) the determined data, which comprise
a vector field $W$ and a positive function $\phi$. Denote by $\nabla ^h$
the covariant derivative with respect to $h$ and by $\Delta _{h}$ the
nonpositive Laplacian on functions, i.e. 
$\Delta _{h}=\divergence _{h} \circ \nabla^h$. 
Let $\D_h$ be the conformal Killing operator relative to $h$,
defined (in index notation) by 
$(\D_h W)_{ab} \definedas \nabla ^h _a W_b + \nabla^h _b W_a -
\frac{2}{n}h _{ab}\nabla^h _m W^m$. The kernel of $\D_h$ consists of the
conformal Killing vector fields on $(M,h)$. Then the system
(\ref{constr1})-(\ref{constr2}) is solvable if and only if for some
choice of conformal background data $(h,\sigma,\tau,\psi,\pi)$ one can solve the
conformally formulated Einstein-scalar field constraint equations
\begin{equation}\label{Hamiltonian}
\begin{split}
\Delta _{h} \phi &-\frac{n-2}{4(n-1)}
\left( R_h - |\nabla ^h \psi|^2_{h}\right)\phi \\
&+
\frac{n-2}{4(n-1)}\left(|\sigma + \D_h W| _{h}^2
  +\pi^2\right)\phi^{-\frac{3n-2}{n-2}}\\
&-\frac{n-2}{4(n-1)}\left(\frac{n-1}{n}\tau^2
-2V(\psi)\right)\phi^{\frac{n+2}{n-2}}
=0,
\end{split}
\end{equation}
\begin{equation}\label{momentum}
\divergence _{h}(\D_h W) 
=
\frac{n-1}{n} \,\phi ^{\frac{2n}{n-2}} \,\nabla ^h \tau - \pi \nabla ^h \psi
\end{equation}
for the determined data $W$ and $\phi>0$, in which case the initial data 
\begin{equation}\label{reconstructed}
\begin{split}
\bar h &=\phi ^{\frac{4}{n-2}} h, \\ 
\bar K &= \phi ^{-2}(\sigma+\D_h W)+\frac{\tau}{n}\,\phi ^{\frac{4}{n-2}} h,\\
\bar \psi &=\psi,\\
\bar \pi &=\phi ^{-\frac{2n}{n-2}}\pi
\end{split}
\end{equation}
solve the original Einstein-scalar field constraint equations.

If one chooses to work under the CMC assumption $\tau=\const$,  
the system (\ref{Hamiltonian})-(\ref{momentum}) becomes
semi-decoupled, which means that the
conformally formulated momentum constraint (\ref{momentum}) becomes a
linear, elliptic, vector equation for $W$ in which the unknown $\phi$
does not appear. If it can be solved, the problem reduces to finding a
positive solution $\phi$ to the conformally formulated Hamiltonian
constraint (\ref{Hamiltonian}), which is commonly referred to as the
Einstein-scalar field Lichnerowicz equation.

In the sequel we will repeatedly use the fact that the
Einstein-scalar field Lichnerowicz equation is conformally covariant
in the following sense. The function $\phi>0$ is a solution to (\ref{Hamiltonian})
for the Einstein-scalar field conformal background data $(h, \sigma, \tau, \psi,
\pi)$, where the vector field $W$ solves (\ref{momentum}) with
respect to $(h, \sigma, \tau, \psi, \pi)$,  if and only if $\theta ^{-1} \phi$
is a solution to (\ref{Hamiltonian}) with respect to the conformally
transformed background data set
\begin{equation*}
(\widetilde h, \widetilde \sigma, \widetilde \tau, \widetilde \psi,
\widetilde \pi) 
\definedas 
(\theta ^{\frac{4}{n-2}}h, \theta^{-2}\sigma,\tau,\psi,\theta ^{-\frac{2n}{n-2}}\pi),
\end{equation*}
where the vector field $\widetilde W$ solves (\ref{momentum}) with
respect to $(\widetilde h, \widetilde \sigma, \widetilde \tau,
\widetilde \psi, \widetilde \pi)$ \cite{Ch-Br-Is-Po}.

\subsection{The Einstein-scalar field constraint equations on closed manifolds}\label{closed}

It should be emphasized that the current work is largely inspired by
\cite{Ch-Br-Is-Po}.  Below we give a brief overview of the method for 
analyzing the Einstein-scalar field constraint equations on closed
(compact without boundary) manifolds developed in that paper.

The authors work in the CMC setting and assume that all
conformal background data sets $(h,\sigma,\tau,\psi,\pi)$ are
smooth. 
In this case, a
smooth solution of the equation (\ref{momentum}) exists provided
that the right hand side 
$\pi \nabla ^h \psi$ is
orthogonal to the space of conformal Killing vector fields on 
$(M, h)$ and is unique if this space is empty.

In order to analyze the Einstein-scalar field Lichnerowicz equation
(\ref{Hamiltonian}), the authors write it in the form
\begin{equation*}
\Delta _{h} \phi - \R_{h,\psi}\phi 
+ \A_{h,W,\pi}\phi^{-\frac{3n-2}{n-2}}
- \B_{\tau,\psi}\phi^{\frac{n+2}{n-2}}
=0,
\end{equation*}
and divide the background data sets into subclasses depending on the
possible signs for the coefficients
\begin{equation*}
\R _{h,\psi}\definedas\frac{n-2}{4(n-1)}(R_h - |\nabla ^h \psi|^2
_{h}), \hspace{0.5cm} \A_{h, W,
  \pi}\definedas\frac{n-2}{4(n-1)}(|\sigma+\D_h W| _{h} ^2+\pi^2),
\end{equation*}
and
\begin{equation*}
\B _{\tau,\psi}\definedas\frac{n-2}{4(n-1)}\left(\frac{n-1}{n}\tau ^2-2V(\psi)\right).
\end{equation*}

In view of the conformal covariance, this splitting of the conformal 
background data set is convenient. Indeed, it was shown in \cite{Ch-Br-Is-Po}
that there always exists a smooth $\theta >0$ such that
$\R_{\widetilde h,\widetilde \psi}$, computed  with respect to the
conformally transformed metric $\widetilde h = \theta ^{\frac{4}{n-2}}
h$ is either positive, negative, or identically zero. As for $\B
_{\tau, \psi}$, there are six different possibilities, namely, this
coefficient can be strictly positive, greater than or equal to zero,
identically zero, less than or equal to zero, strictly negative, or of
indeterminate sign. This combined with the two options $\A_{h, W,
  \pi}\equiv 0$ and $\A_{h, W, \pi}\not \equiv 0$ gives rise to 36
classes of Einstein-scalar field CMC conformal background data $(h,
\sigma, \tau, \psi, \pi)$.

Concluding this brief overview of \cite{Ch-Br-Is-Po}, we note that for
many of the classes it was possible to determine whether or not the
smooth positive solution exists. More details are to be found in the
original paper.

\subsection{Asymptotically hyperbolic geometries}\label{as-hyp} 

The goal of this paper is to find asymptotically hyperbolic solutions
to (\ref{constr1})-(\ref{constr2}). This is done by solving  
(\ref{Hamiltonian})-(\ref{momentum}) with asymptotically hyperbolic
conformal background data and correct asymptotics of the solutions.  

The prototype for asymptotically hyperbolic manifolds is a constant
negative curvature hyperboloid in Minkowski spacetime. They are the
interiors of compact manifolds with boundary; the points on the
boundary represent ``points at infinity'' for the asymptotically
hyperbolic manifold.

\begin{definition}\cite{An-Cr-Fr}\label{hyperboloidal}
Let $(M,g)$ denote an oriented, compact $C^ \infty$ Riemannian
manifold of dimension $n \geq 3$, with nonempty boundary $\partial M$
and interior $\widetilde{M}$. Assume that $\rho \in C^ \infty (M)$ is
a defining function for $\partial M$, i.e. $\rho > 0$ on
$\widetilde{M}$ while $\rho =0$ but $d \rho \neq 0$ everywhere on
$\partial M$. Then the manifold $(\widetilde{M}, h)$, where $h=\rho
^{-2} g$, is said to be conformally compact. If, in addition, $|d
\rho| _g=1$ holds on $\partial M$, then $(\widetilde{M}, h)$ is called
asymptotically hyperbolic.
\end{definition}

A standard calculation shows that the sectional curvature $K_h$ of a
conformally compact manifold $(\widetilde M,h)$ satisfies $K_h (p)
\rightarrow -|d \rho |_g^2 (q)$ as $p \rightarrow q \in \partial M$,
which means that if $(\widetilde M,h)$ is asymptotically hyperbolic
then all sectional curvatures tend to $-1$ at infinity.

In this paper we will largely rely upon the fact that any
asymptotically hyperbolic geometry $(\widetilde{M}, h)$ is conformally
related to one with constant negative scalar curvature
\cite{An-Cr-Fr}. Namely, on every conformally compact manifold
$(\widetilde{M}, h)$ there exists a unique function  $w$ such that $w>0$ on $M$ such
that the metric $\widetilde h=\left(w \rho
  ^{\frac{n-2}{2}}\right)^{\frac{4}{n-2}}h=\rho^{-2}\left(w\rho^{\frac{n-2}{2}}\right)^{\frac{4}{n-2}}g$ has scalar curvature
$-n(n-1)$. In addition, on an  asymptotically hyperbolic manifold
$(\widetilde{M}, h)$ the conformal factor  $w \rho ^{\frac{n-2}{2}}$
satisfies $w \rho ^{\frac{n-2}{2}}\rightarrow 1$ as $\rho \rightarrow
0$, which implies that $(\widetilde M, \widetilde h)$ is also
asymptotically hyperbolic.

Another important technical lemma to be used in this work is the
following version of the maximum principle.

\begin{theorem}\cite[Theorem 3.5]{Gr-Lee}  \label{max}
Suppose that $(\widetilde M, h)$ is an asymptotically hyperbolic
manifold and $f \in C^2 (M)$ is bounded. Then there exists a
sequence $x_k \in M$ such that
\begin{itemize} 
\item[(i)] $\lim _{k \rightarrow \infty} f(x_k)=\inf _M f$;

\item[(ii)] $\lim _{k \rightarrow \infty} |\nabla ^h f(x_k)|_h=0$;

\item[(iii)] $\liminf _{k \rightarrow \infty} \Delta _h f(x_k)\geq 0$.
\end{itemize}  
\end{theorem}

When dealing with asymptotically hyperbolic manifolds, it is important
to control the behavior of the data near $\partial M$, which is achieved by using weighted spaces with weights being the powers of the defining function $\rho$. The simplest example of weighted space is $C^0_\delta (M)$, $\delta \in \rn$. Its elements are continuous functions such that $\|u\|_{C^0_\delta(M)}\definedas \sup _M |\rho^{-\delta}u|$ is finite. Using the standard definitions of $W^{k,p}(\widetilde{M},h)$ and $C^{k,\alpha}(\widetilde{M},h)$ (see e.g. \cite{Lee}) for each real number $\delta$ one can also define the weighted Sobolev spaces $W_\delta^{k,p}(M)\definedas \{\,\rho^\delta u\,: \,u\in W^{k,p}(\widetilde{M},h)\,\}$
and the weighted H\"older spaces $C_\delta^{k,\alpha}(M)\definedas \{\,\rho^\delta u\,: \,u\in C^{k,\alpha}(\widetilde{M},h)\,\}$ with the respective norms given by $\|u\|_{W^{k,p}_\delta(M)}\definedas \|\rho^{-\delta}u\|_{W^{k,p}(\widetilde{M},h)}$ and $\|u\|_{C^{k,\alpha}_\delta (M)}\definedas \|\rho^{-\delta}u\|_{C^{k,\alpha}(\widetilde{M},h)}$. All results pertaining to these spaces and elliptic operators on asymptotically hyperbolic manifolds to be used in this paper can be found in \cite{Lee}. 

\section{Main results}

From now on let $(\widetilde M, h)$ be an asymptotically hyperbolic
manifold as in Definition \ref{hyperboloidal} and suppose that $\dim
M=n$. Our goal is to analyze the solvability of the conformally
formulated Einstein-scalar field constraint equations
(\ref{Hamiltonian})-(\ref{momentum}) on $(\widetilde M, h)$, for which
we will employ essentially the same strategy as in
\cite{Ch-Br-Is-Po}. In particular, we will assume that all conformal
background data $(h, \sigma, \tau, \psi, \pi)$ are smooth and that
$\tau$ is constant.

\subsection{Solving the momentum constraint}

By analogy with \cite{Ch-Br-Is-Po}, before formulating our results for
the full set of constraint equations we restrict ourselves to
considering only those sets of conformal background data  for which
the conformally formulated momentum constraint (\ref{momentum}) is
solvable. Due to the CMC assumption and
the following result many such sets of conformal background data can be found. 

\begin{proposition}\cite[Proposition G]{Lee}\label{Fredholm}
For an integer $k>0$ and real $\alpha$ and $\delta$ such that
$0<\alpha<1$ and $0<\delta<n$,  the vector  Laplacian $\divergence
_{h} \circ \D_h : C_{\delta} ^{k+2,\alpha}(M,h) \rightarrow C_{\delta}
^{k,\alpha}(M,h) $  is a Fredholm operator. Its index is zero, and its kernel is equal to the $L^2$ kernel of $\divergence
_{h} \circ \D_h $.
\end{proposition}

Since an asymptotically hyperbolic manifold does not have any conformal 
Killing vector fields in $L^2 (M,h)$ \cite[Lemma 6.7]{Romain},
the operator in the above proposition is an isomorphism.

\subsection{Solving the Hamiltonian constraint}

Our next goal is to determine for which choices of the Einstein-scalar
field conformal data $(h, \sigma, \tau, \psi, \pi)$ such that
$(\widetilde M, h)$ is asymptotically hyperbolic the Lichnerowicz
equation with a vector field $W$,
\begin{equation}\label{LE}
\Delta _h \phi-\R_{h,\psi}\phi
+\A_{h,W,\pi}\phi^{-\frac{3n-2}{n-2}}-\B_{\tau,\psi}\phi^{\frac{n+2}{n-2}}
=0,
\end{equation} 
admits a smooth solution $\phi >0$ which satisfies the boundary
condition
\begin{equation}\label{bd}
\phi \rightarrow 1 \text{ as } \rho \rightarrow 0. 
\end{equation}
This boundary condition guarantees that 
$\bar h = \phi^{\frac{4}{n-2}} h=\rho^{-2}\phi^{\frac{4}{n-2}}g$ is also asymptotically hyperbolic.

The conformal covariance property and the results recalled
in Section \ref{as-hyp} allow us to assume that $R_h = -n(n-1)$, which
we will do from now on. In addition, since the mean curvature of the
$n$-dimensional constant negative curvature hyperboloid in Minkowski
spacetime is equal to $n$, the same value will be assumed for the
constant mean curvature: $\tau = n$. This gives the following
expressions for the coefficients of (\ref{LE}): 
\begin{equation*}
\R _{h,\psi}=\frac{n-2}{4(n-1)}(-n(n-1)-|\nabla ^h\psi|^2 _{h}), 
\hspace{0.5cm} 
\A_{h, W, \pi}=\frac{n-2}{4(n-1)}(|\sigma+\D_h W| _h ^2+\pi^2),
\end{equation*}
and
\begin{equation*}
\B _{\tau,\psi}=\frac{n-2}{4(n-1)}\left(n(n-1)-2V(\psi)\right).
\end{equation*}
In order to analyze (\ref{LE}) we split the set of background data
into subclasses in the same way as it was done for closed manifolds
(see Section \ref{closed}). Since  $\R _{h, \psi}$ is strictly
negative, we see that the possibilities we have to analyze are the
same as those listed in the first rows of Table 1 and Table 2 in
\cite{Ch-Br-Is-Po}. Anticipating the theorems to be formulated below
we note that our results bear a lot of similarity  with those listed
in the aforementioned tables. In particular, we will show 
that if the potential $V$  is such that $\B _{\tau,\psi}$ is
non-positive or zero, then (\ref{LE})-(\ref{bd}) admits no solution. We will also
see that, under reasonable restrictions on the conformal background
data, the condition $\B _{\tau,\psi}>0$ guarantees the solvability of the Einstein-scalar field
Lichnerowicz equation (\ref{LE}) with the boundary condition (\ref{bd}), while
in the case $\B _{\tau,\psi}\geq 0$ a partial result can be proved.
 
\subsection{The main theorems}

In this paper we prove the following two
theorems. The first one is a non-existence result.

\begin{theorem}\label{non-existence}
Assume that we are given a manifold $M$ and conformal background data
$(h,\sigma,\tau,\psi,\pi)$,  with $\tau=n$ on $M$, such that
$(\widetilde M,h)$ is asymptotically hyperbolic. If   
\begin{equation*}
\inf V(\psi) \geq \frac{n(n-1)}{2},
\end{equation*}
then there is no solution to (\ref{LE})-(\ref{bd}).
\end{theorem}

In fact, the argument in the proof shows that there can be no solution to 
(\ref{LE}) which is bounded from below by a positive constant. 

The second theorem is an existence result.

\begin{theorem}\label{existence}
Assume that we are given a manifold $M$ and conformal background data
$(h,\sigma,\tau,\psi,\pi)$  with $\tau=n$ on $M$ such that
$(\widetilde M,h)$ is asymptotically hyperbolic. Suppose that for some
$0<\delta<n-1$ we have 
\begin{equation} \label{exist_ass_1}
2V(\psi)+|\nabla ^h \psi|^2 _{h} \in C^0 _{\delta},
\end{equation}
and the conformally reformulated momentum
constraint (\ref{momentum}) is solvable with the solution $W$
satisfying 
\begin{equation} \label{exist_ass_2}
|\sigma+\D_h W| _h ^2 + \pi^2 \in C^0 _{\delta}.
\end{equation}
If the potential $V$ is bounded from below and satisfies 
\begin{equation} \label{exist_ass_3}
\sup V(\psi) < \frac{n(n-1)}{2},
\end{equation}
then there is a unique positive smooth solution $\phi$ to the
Einstein-scalar field Lichnerowicz equation (\ref{LE}) such that $\phi
-1 \in C^0 _{\delta}$ and the
initial data $(\bar h, \bar K, \bar \psi, \bar \pi)$  
defined in the equations (\ref{reconstructed}) satisfy the
Einstein-scalar field constraint equations
(\ref{constr1})-(\ref{constr2}).
\end{theorem}

\section{Proof of Theorem \ref{non-existence}}

Note that the assumption made on $V$ implies that $\B_{\tau,\psi}\leq
0$. Using this we will prove that the
Einstein-scalar field Lichnerowicz equation (\ref{LE}) with the boundary condition (\ref{bd})
admits no positive smooth solution. Assume that smooth $\phi > 0$ satisfies (\ref{LE})-(\ref{bd}) and set
\begin{equation*}
\alpha = \inf _M \phi.
\end{equation*} 

If $\alpha$ is attained at some point $p \in \widetilde{M}$ then
$\alpha$ is strictly positive. Applying the maximum principle, one
immediately gets a contradiction.
 
Now suppose that $\alpha$ is not attained in $\widetilde{M}$. In this
case $\alpha = 1$ by (\ref{bd}). Moreover, by Theorem \ref{max}, there
exists a sequence of points $p_k \in \widetilde{M}$, such that 
$p_k \rightarrow p \in \partial M$, and
\begin{equation*}
\phi (p_k) \rightarrow \alpha, 
\hspace{20pt} 
\liminf _{k \rightarrow \infty} \Delta _{h} \phi (p_k) \geq 0.
\end{equation*}
Evaluating (\ref{LE}) at $p_k$ and then passing to limit when $k\rightarrow \infty$ yields
\begin{equation*}
\limsup_{k\rightarrow \infty}\left[- \R_{h,\psi}(p_k)\alpha + \A_{h,W,\pi}(p_k)\alpha^{-\frac{3n-2}{n-2}} 
- \B_{\tau,\psi}(p_k)\alpha^{\frac{n+2}{n-2}}\right]
\leq 0,
\end{equation*}
which is a contradiction.

\section{Sub- and supersolution method} 

The proof of the existence result relies on the method of sub- and 
supersolutions. In the proposition below (which was stated and proved in a more general form in \cite{An-Cr}) we recall how to construct sub- and supersolutions on asymptotically hyperbolic manifolds.

\begin{proposition}\cite{An-Cr}\label{subsuperconstruction}
Let $(M,h)$ be an asymptotically hyperbolic manifold. Consider the equation 
\begin{equation}\label{general equation}
\Delta _h u+F(x,u)=0
\end{equation} 
for a scalar function $u$ on $M$.
Suppose that 
\begin{itemize}
\item [(i)] There exist constants $C_+\geq 1$ and $C_-\leq 1$ such that for any $x\in M$ we have
\begin{equation*}
F(x,C_+)\leq 0 \text{ and } F(x,C_-)\geq 0. 
\end{equation*}
\item[(ii)] There exist a constant $C>0$ and $0<\delta<n-1$ such that for $0<\rho<\rho_0$ we have 
\begin{equation*}
\begin{aligned}
\text{if } 1\leq u\leq C_+ &\text{ then } F(x,u)\leq C\rho^{\delta};\\
\text{if } C_-\leq u\leq 1 &\text{ then } F(x,u)\geq -C\rho^{\delta}.
\end{aligned}
\end{equation*}
\end{itemize}
Then there exists a constant $B>0$ such that $u_+=\min\{1+B\rho^\delta,C_+\}$ and $u_-=\max\{1-B \rho^\delta,C_-\}$ are respectively a supersolution and a subsolution of the equation (\ref{general equation}), i.e. 
\begin{equation*}
\Delta _h u_++F(x,u_+)\leq 0 \text{ and } \Delta _h u_-+F(x,u_-) \geq 0.
\end{equation*}

\end{proposition}
\begin{proof}
A computation shows that
\begin{equation}\label{laplacian}
\Delta_h \rho ^{\delta}=\delta(1-n+\delta)\rho ^{\delta}|\nabla ^g \rho|_g ^2+\delta \rho ^{\delta+1} \Delta _g \rho,
\end{equation}
which means that there exists $C_1>0$ and $\rho_1>0$ such that $\Delta _h \rho^\delta \leq -C_1 \rho^\delta$ for $0<\rho<\rho_1$. Let $\rho'\definedas \min\{\rho_0,\rho_1\}$. 

We construct a constant $B_+>0$ such that $u_+=\min\{1+B_+\rho^\delta,C_+\}$ is a supersolution. First, we want $u_+=C_+$ for all $\rho\geq\rho'$, and therefore we require $B_+\geq \frac{C_+-1}{(\rho')^\delta}$. Second, if $u_+=C_+$ then it is clearly a supersolution. Now it remains to ensure that $\Delta_h u_++F(x,u_+)\leq 0$ on that part of $\{0<\rho<\rho'\}$ where $u_+=1+B_+\rho^\delta$. But for $0<\rho<\rho'$ we have 
\begin{equation*}
\Delta_h u_++F(x,u_+)\leq (-B_+C_1+C)\rho^\delta. 
\end{equation*}
Consequently, the constant $B_+=\max \{\frac{C}{C_1},\frac{C_+-1}{(\rho')^\delta} \}$ satisfies our needs. 

The constant $B_-$ such that $u_-=\max\{1+B_-\rho^\delta,C_-\}$ is a subsolution is constructed similarly. Finally, we set $B\definedas\max\{B_+,B_-\}$.
\end{proof} 

Next statement is a sub- and supersolution theorem for asymptotically hyperbolic manifolds. 

\begin{theorem}\label{subsupertheorem}
Let $(M,h)$ be an asymptotically hyperbolic manifold and suppose that $F:M\times(0,+\infty)\rightarrow \rn$ is smooth in both arguments. Assume that there exist continuous functions $u_-$ and $u_+$ in $W^{1,2}_{loc}(M)$  such that $0<u_-<u_+<C$ and $u_-$ and $u_+$ are respectively a weak subsolution and a weak supersolution of (\ref{general equation}). Then there exists a smooth solution $u$ of (\ref{general equation}) on $M$ such that $u_-\leq u\leq u_+$.
\end{theorem}

\begin{remark}
If $F:M\times[0,+\infty)\rightarrow \rn$ is smooth in both arguments, then $u_-=0$ can also be used as a subsolution.
\end{remark}

\begin{proof}
The proof is standard, cf. \cite[Proposition 2.1]{Av-McO}. 

To construct a solution of (\ref{general equation}) on $M$, suppose that 
\begin{equation*}
M=\bigcup _{k=1} ^{\infty} \Omega _k \, ,
\end{equation*}
where $\Omega _k$ are open and bounded with $C^1$ boundary and $\overline{\Omega}_k\subset
\Omega_{k+1}$. It is easy to check that maximum principle holds for functions in $W^{1,2}(\Omega _k)\cap C(\overline{\Omega}_k)$, therefore the 
monotone iteration scheme \cite[Theorem 2.3.1]{Sat} can be applied to produce a $W^{2,p}$ solution $u_k$ such that $u_-\leq u_k\leq u_+$ on $\Omega _k$. 

Let us consider the sequence $\{u _k\}_{k>3}$ on $\overline{\Omega}
_3$.  By construction,  for $x \in\overline{\Omega}_3$ and $k \geq 4$, using local Schauder estimates
we find that

\begin{equation*}
\left\|u _k\right\|_{W^{2,p} (\Omega_2)}
\leq 
C \left\|F(\cdot,u_k(\cdot))\right\|_{L^p(\Omega_3)}
+ C \left\|u_k\right\|_{L^p(\Omega_3)}<C,
\end{equation*}
where the generic constant $C$ does not depend on $k$ and $p \geq 1$ is arbitrary.
Suppose that $p>n$. Then it follows from the Sobolev embedding theorem
that $\left\|u_k\right\|_{C^{0,\gamma}(\overline{\Omega}_2)}\leq
C$. We apply interior elliptic estimates, and deduce that
\begin{equation*}
\left\|u_k\right\|_{C^{2,\gamma} (\Omega_1)}
\leq C \left\|F(\cdot,u_k(\cdot))\right\|_{C^{0,\gamma}(\Omega_2)}+C \left\|u_k\right\|_{C^{0,\gamma}(\Omega_2)}<C
\end{equation*}
uniformly in $k$. Since $C^{2,\gamma}(\Omega _1)\subset\subset C^2
(\Omega _1)$, we finally deduce that $\{u_k\}$ has a subsequence
$\{u_{k_i}\}$ which converges to a solution of (\ref{general equation}) on
$\Omega _1$. 

We set $u^1_i \definedas u_{k_i}$.
We repeat this procedure with $\{u^1_{k}\}$ on
$\overline{\Omega} _4$ to obtain a subsequence $\{u^1_{k_i}\}$ which
converges to a solution of (\ref{general equation}) on $\Omega _2$. Set 
$u^2_i \definedas u^1_{k_i}$. Proceeding
by induction for every $j$ we can construct a subsequence
$\{u^j_{k_i}\}$ converging to a solution of (\ref{general equation}) on
$\Omega_{j+1}$. Then a diagonal subsequence $\{u ^j _{k_j}\}$
converges to a $C^2$ solution $u$ of (\ref{general equation}) on $M$. Further regularity of $u$ follows by induction and bootstrap argument.
\end{proof}

\section{Existence} 

In this section we prove the existence part of Theorem \ref{existence}. For the sake of convenience, we state this result in the following form. 

\begin{proposition}
If $\B_{\tau,\psi}$ is positive and is bounded from above, and if for some $0<\delta <n-1$ we have
$\A_{h,W,\pi}\in C^0 _{\delta}$ and $\R_{h,\psi}+\B_{\tau, \psi}\in
C^0 _{\delta}$, then the Einstein-scalar field Lichnerowicz equation
(\ref{LE}) admits a positive smooth solution $\phi$ such that
$\phi-1\in C^0 _{\delta}$. 
\end{proposition}

Note that the assumption (\ref{exist_ass_3}) on $V$ implies 
$\B_{\tau,\psi}> 0$, that $V$ is bounded from below yields that $\B_{\tau,\psi}$ is bounded from above, and that from 
(\ref{exist_ass_1}) it follows that 
$\R_{h,\psi}+\B_{\tau, \psi}\in C^0 _{\delta}$. Moreover, the assumption 
(\ref{exist_ass_2}) is exactly that $\A_{h,W,\pi}\in C^0 _{\delta}$.

\begin{proof}
We first apply Proposition \ref{subsuperconstruction} with 
\begin{equation*}
F(x,u)\definedas-\R_{h,\psi}u
+\A_{h,W,\pi}u^{-\frac{3n-2}{n-2}}-\B_{\tau,\psi}u^{\frac{n+2}{n-2}}
\end{equation*}
 in order to construct sub- and supersolutions of (\ref{LE}). 

It is readily checked that $F(x,C_-)\geq 0$ is satisfied provided that $C_-^{\frac{4}{n-2}}\leq\frac{\inf(-\R_{h,\psi})}{\sup \B_{\tau,\psi}}$. It is also easy to see that $\lim _{u\rightarrow+\infty}\sup_x F(x,u)=-\infty$, hence there exists $C_+\geq 1$ such that $F(x,C_+)\leq 0$.

Moreover,  if $1\leq u\leq C_+$ then 
\begin{equation*}
F(x,u)\leq -u(\R_{h,\psi}+\B_{\tau, \psi})+\A_{h,W,\pi}\leq C \rho^\delta,
\end{equation*} 
by our assumptions on $\A_{h,W,\pi}$, $\R_{h,\psi}$ and $\B_{\tau, \psi}$. Similarly, for $C_-\leq u\leq 1$ we have 
\begin{equation*}
F(x,u)\geq -u(\R_{h,\psi}+\B_{\tau, \psi})+\A_{h,W,\pi}\geq -C \rho^\delta.
\end{equation*} 

By Proposition \ref{subsuperconstruction} there exists a constant $B>0$ such that $u_+=\min\{1+B\rho^\delta,C_+\}$ and $u_-=\max\{1-B\rho^\delta,C_-\}$ are respectively a supersolution and a subsolution of (\ref{LE}). Since $u_+$, $u_-$ and $F$ satisfy the conditions of Theorem \ref{subsupertheorem} we deduce that there exists a smooth positive solution $u$ of (\ref{LE}) such that $u-1\in C^0 _{\delta}$. \end{proof}

\section{Partial result}

When $\B_{\tau,\psi}\geq 0$ is not strictly positive, Proposition \ref{subsuperconstruction} no longer applies, since the constant $C_+$ might not exist. The theorem below is aimed at facilitating the analysis of solvability of the Einstein-scalar field Lichnerowicz equation in this case. Namely, it shows that the situation when $\A_{h,W,\pi}\not\equiv 0$ can be reduced to the case $\A_{h,W,\pi}\equiv 0$, when the problem is that of prescribed $\R_{h,\psi}$. 

\begin{theorem}
If the coefficients of the Einstein-scalar field Lichnerowicz equation (\ref{LE}) are such that $\B_{\tau,\psi}$ is nonnegative and is bounded from above, $\A_{h,W,\pi}\in C^{\alpha}_{\delta}$ and $\B_{\tau, \psi}+\R_{h,\psi}\in C^0_\delta$ then the following statements are equivalent
\begin{itemize}
\item [(i)] The Einstein-scalar field Lichnerowicz equation (\ref{LE}) admits a smooth positive solution $\phi$ such that $\phi-1\in C_{\delta}^0$.
\item [(ii)] The Einstein-scalar field Lichnerowicz equation (\ref{LE}) with $\A_{h,W,\pi}\equiv 0$ admits a smooth positive solution $\phi$ such that $\phi-1\in C_{\delta}^0$.
\item [(iii)] There exists a smooth positive $\phi$ such that $\phi-1\in C_{\delta}^0$ and $\widetilde{h}\definedas \phi^{\frac{4}{n-2}}h$ satisfies 
$\R_{\widetilde{h},\widetilde{\psi}}=-\B_{\tau,\psi}$.
\end{itemize} 
\end{theorem}
\begin{proof}
If $\widetilde{h}= \phi^{\frac{4}{n-2}}h$ then
\begin{equation*}
\R_{\widetilde h, \widetilde \psi}=\phi ^{-\frac{n+2}{n-2}}
(-\Delta _h \phi+\R_{h,\psi} \phi).
\end{equation*} 
It is obvious from this formula that (ii) and (iii) are equivalent. 

Suppose that (i) holds. Since $\A_{h,W,\pi}\geq 0$, the solution of (\ref{LE}) is also a supersolution of (\ref{LE}) with $\A_{h,W,\pi}\equiv 0$, and from now on it will be denoted by $\phi_+$. Note that there exists $C_0>0$ such that $\phi_+\geq 1-C_0 \rho^\delta$. 

The respective subsolution $\phi_-$ is easily constructed by Proposition \ref{subsuperconstruction}. Namely, if $\B_{\tau,\psi} \equiv 0$, we pick a constant $C_-$ so that $C_-\leq \min \phi_+$, and if $\B_{\tau,\psi} \not\equiv 0$, we choose
\begin{equation*}
C_-\leq \min \left\{ \left( \frac{\inf (-\R_{h,\psi})}{\sup \B_{\tau,\psi}} \right)^{\frac{n-2}{4}}, \min \phi_+ \right\}. 
\end{equation*}
It is also easy to check that if $C_-\leq u \leq 1$ then $-\R_{h,\psi}u-\B_{\tau,\psi}u^{\frac{n+2}{n-2}}\geq -C\rho^{\delta}$. By Proposition  \ref{subsuperconstruction} we deduce that there exists $B>0$ such that $\phi_-=\max\{C_-,1-B\rho^{\delta}\}$ is a subsolution. However, note that we might need to increase $B$ in order to ensure that $\phi_-\leq 1-B\rho^{\delta}\leq 1-C_0 \rho^\delta \leq \phi_+$. Applying Theorem \ref{subsupertheorem}, we deduce that (ii) holds.

The proof will be completed if we show that (iii) implies (i). The following supplementary lemma will be required.
\begin{lemma}\label{help-lemma} 
Suppose that smooth functions $f$ and $\xi\not\equiv 0$ are nonnegative, and, moreover, that $\xi \in C^\alpha _{\delta}$ for some
$0<\delta<n-1$ and $0<\alpha<1$. Then the equation
\begin{equation}\label{help}
-\Delta_h u+ f u=\xi
\end{equation}
has a nonnegative smooth solution $u \in C^0 _{\delta}$. 
\end{lemma}

\begin{proof}
The proof is again based on sub- and supersolution method. It is clear that $u_-=0$ can be used as a subsolution, thus it remains to construct a supersolution.

Since the $L^2$ kernel of $\Delta_h$ is zero, by Theorem C and Theorem F in \cite{Lee}, we see that there exists a solution $v\in C_{\delta}^{2,\alpha}$ of the equation $-\Delta_h v=\xi$ provided that $0<\delta<n-1$ and $0<\alpha<1$. Moreover, $v$ is not constant  since $\xi \not \equiv 0$, and from the Hopf strong maximum principle it follows that $v$ is nonnegative.  Since $f$ is nonnegative, it is clear that $u_+\definedas v$ can be used as a supersolution. 

The application of Theorem \ref{subsupertheorem} completes the proof. \end{proof}

Assume that (iii) holds, that is, that there exists a smooth positive $\phi_1$ such that $\phi_1-1\in C_{\delta}^0$ and $\widetilde{h}\definedas \phi_1^{\frac{4}{n-2}}h$ satisfies 
\begin{equation}\label{prescribed1}
\R_{\widetilde{h},\widetilde{\psi}}=-\B_{\tau,\psi}.
\end{equation}
To show that (i) holds we will follow the argument from the proof of Proposition 3 in \cite{Ch-Br-Is-Po}, which is based on a method presented by Maxwell in \cite{Max}. Since we have already shown that (ii) and (iii) are equivalent, it can be assumed that $\A_{h,W,\pi}\not \equiv 0$ and that 
\begin{equation}\label{prescribed2}
\Delta _h \phi_1-\R_{h,\psi}\phi_1-\B_{\tau,\psi}\phi_1^{\frac{n+2}{n-2}}=0.
\end{equation}

Note that from (iii) and the fact that $\R_{h,\psi}+\B_{\tau, \psi}\in C^0 _{\delta}$ it can be deduced that $\phi_1\in C^{0,\gamma}_0$ for some $0<\gamma<1$. Indeed, we know that $\phi_1=1+w$ where $w\in C_\delta ^0$. It is obvious that $w$ satisfies
\begin{equation*}
\Delta _h w=\phi_1(\R_{h,\psi}+\B_{\tau,\psi}\phi_1^{\frac{4}{n-2}}).
\end{equation*}
Observe that $\phi_1(\R_{h,\psi}+\B_{\tau,\psi}\phi_1^{\frac{4}{n-2}})=\phi_1(\R_{h,\psi}+\B_{\tau,\psi}+O(\rho^\delta))$ is in $C^0_\delta$, hence in $W_{\delta'}^{0,p}$ for $p$ sufficiently large,  where $0<\delta'<\delta-\frac{n-1}{p}$. It  will be assumed that $p>n$. Since $0<\delta'+\frac{n-1}{p}<\delta< n-1$, we deduce by Theorem C and Theorem F in \cite{Lee} that $w\in W^{2,p}_{\delta'}$, hence $w\in C^{1,\gamma}_{\delta'}$ for some $0<\gamma<1$. Finally, $w\in C^{0,\gamma}_{0}$ and the same holds for $\phi_1=1+w$.

Recall that $\phi_1$ is bounded away from zero, thus $\A_{\widetilde{h},\widetilde{W},\widetilde{\pi}}=\phi _1^{-\frac{4n}{n-2}} \A_{h,W,\pi} \in C^\lambda _{\delta}$ for $\lambda\definedas\min\{\alpha,\gamma\}$. Since both $\B _{\tau,\psi}$ and $\A_{\widetilde{h},\widetilde{W},\widetilde{\pi}}$ are nonnegative, it follows from Lemma \ref{help-lemma} that the equation
\begin{equation*}
-\Delta _{\widetilde{h}} \theta+\B_{\tau,\psi} \theta=\A_{\widetilde{h},\widetilde{W},\widetilde{\pi}}
\end{equation*} 
has a nonnegative smooth solution $\theta \in C^0 _\delta$.
Hence $\phi_2=1+\theta$ solves the equation
\begin{equation}\label{phi2}
-\Delta _{\widetilde{h}} \phi_2+\B_{\tau,\psi} \phi_2
=\A_{\widetilde{h},\widetilde{W},\widetilde{\pi}}+\B_{\tau,\psi}.
\end{equation}
Set $\widehat{h}\definedas\phi _2 ^\frac{4}{n-2} \widetilde{h}$. Using
(\ref{prescribed1}) and (\ref{phi2}), we compute
\begin{equation*}
\begin{split}
\R_{\widehat{h},\widehat{\psi}}
&=
\frac{n-2}{4(n-1)}(R_{\widehat{h}}-|\nabla ^{\widehat{h}} \psi |_{\widehat{h}}^2)\\
&=
\phi_2^{-\frac{n+2}{n-2}}(-\Delta _{\widetilde{h}} \phi_2
+\R_{\widetilde{h},\widetilde{\psi}}\phi_2)\\ 
&=\phi_2
^{-\frac{n+2}{n-2}}(-\Delta _{\widetilde{h}} \phi_2
-\B_{\tau,\psi}\phi_2)\\ 
&=\phi_2 ^{-\frac{n+2}{n-2}}(\A
_{\widetilde{h},\widetilde{W},\widetilde{\pi}} +\B_{\tau,\psi}(1-2
\phi_2))\\ 
&=\A_{\widehat{h},\widehat{W},\widehat{\pi}} \phi _2
^{\frac{3n-2}{n-2}}+\B_{\tau,\psi}(\phi _2 ^{-\frac{n+2}{n-2}}-2\phi
_2 ^{-\frac{4}{n-2}}),
\end{split}
\end{equation*}
and the Einstein-scalar field Lichnerowicz equation with respect to
the conformally transformed background data $(\widehat{h},
\widehat{\sigma},\widehat{\tau},\widehat{\psi},\widehat{\pi})$ becomes 
\begin{equation*}
\Delta _{\widehat{h}} \phi +
\A_{\widehat{h},\widehat{W},\widehat{\pi}} (\phi
^{-\frac{3n-2}{n-2}}-\phi _2 ^{\frac{3n-2}{n-2}} \phi)+\B _{\tau,
  \psi} (2\phi _2 ^{-\frac{4}{n-2}}\phi-\phi _2
^{-\frac{n+2}{n-2}}\phi-\phi ^{\frac{n+2}{n-2}})=0.
\end{equation*}

It is checked straightforwardly that 
\begin{equation*}
F(x,u)\definedas
\A_{\widehat{h},\widehat{W},\widehat{\pi}} (u
^{-\frac{3n-2}{n-2}}-\phi _2 ^{\frac{3n-2}{n-2}} u)+\B _{\tau,
  \psi} (2\phi _2 ^{-\frac{4}{n-2}}u-\phi _2
^{-\frac{n+2}{n-2}}u-u ^{\frac{n+2}{n-2}})
\end{equation*}
satisfies the conditions of Proposition \ref{subsuperconstruction}. Indeed, one easily verifies that $C_+\geq 2^\frac{n-2}{4}$ and $C_-\leq \inf \phi_2^{-1}$ satisfy $F(x,C_+)\leq 0$ and $F(x,C_-)\geq 0$ respectively. Moreover, if $1\leq u\leq C_+$ then 
\begin{equation*}
u^{-\frac{3n-2}{n-2}}-\phi _2 ^{\frac{3n-2}{n-2}} u\leq 1-\phi_2^{\frac{3n-2}{n-2}}=-\frac{3n-2}{n-2}\theta+o(\theta)\leq C\rho^{\delta},
\end{equation*}
and 
\begin{equation*}
\begin{split}
2\phi _2 ^{-\frac{4}{n-2}}u-\phi _2^{-\frac{n+2}{n-2}}u-u ^{\frac{n+2}{n-2}}=u \phi_2^{-\frac{n+2}{n-2}}(\phi_2-1)+u(\phi_2^{-\frac{4}{n-2}}-u^{\frac{4}{n-2}})\\
\leq u \phi_2^{-\frac{n+2}{n-2}}\theta\leq C \rho^\delta.
\end{split}
\end{equation*}
since $\phi_2=1+\theta\geq 1$, and $\theta\in C^0_{\delta}$. This implies that $F(x,u)\leq C\rho^{\delta}$ for $1\leq u\leq C_+$, and it is similarly checked that $F(x,u)\geq -C \rho^{\delta}$ for $C_-\leq u\leq 1$. 

By Proposition \ref{subsuperconstruction}  the sub- and supersolutions are now constructed, and it only remains to apply Theorem \ref{subsupertheorem} to complete the proof.
\end{proof}

\section{The uniqueness}

In this final section we prove that if the solution of  the Einstein-scalar field Lichnerowicz equation (\ref{LE}) with the boundary condition (\ref{bd}) exists, then it is unique. By this we, in particular, complete the proof of Theorem \ref{existence}.

Assume that $\phi _1$ and $\phi _2$ are two positive solutions to the
boundary value problem (\ref{LE})-(\ref{bd}). By the conformal
covariance, $1 =\phi _2 \phi _2 ^{-1}$ and 
$\phi \definedas \phi _1 \phi _2 ^{-1}$ are then solutions to the
equation
\begin{equation*}
\Delta _{\widetilde h} \phi
-\R_{\widetilde h,\widetilde \psi}\phi
+\A_{\widetilde h,\widetilde W,\widetilde \pi}\phi^{-\frac{3n-2}{n-2}}
-\B_{\tau,\psi}\phi^{\frac{n+2}{n-2}}
=0
\end{equation*} 
with respect to the conformally transformed metric $\widetilde h
\definedas \phi _2 ^{\frac{4}{n-2}} h$. Hence
\begin{equation*}
\R_{\widetilde h,\widetilde \psi}
=\A_{\widetilde h,\widetilde W,\widetilde \pi}-\B_{\tau,\psi}
\end{equation*}
and $\phi$ satisfies
\begin{equation}\label{transformed}
\Delta _{\widetilde h} \phi
-(\A_{\widetilde h,\widetilde W,\widetilde \pi}
-\B_{\tau,\psi})\phi
+\A_{\widetilde h,\widetilde W,\widetilde \pi}\phi^{-\frac{3n-2}{n-2}}
-\B_{\tau,\psi}\phi^{\frac{n+2}{n-2}}
=0.
\end{equation}

Set 
\begin{equation*}
\alpha\definedas\inf _M \phi.
\end{equation*}
 If $\alpha$ is achieved at $p \in \widetilde M$ then  
\begin{equation}\label{impossible}
-(\A_{\widetilde h,\widetilde W,\widetilde \pi}(p)
-\B_{\tau,\psi}(p))\alpha
+\A_{\widetilde h,\widetilde W,\widetilde
  \pi}(p)\alpha^{-\frac{3n-2}{n-2}}
-\B_{\tau,\psi}(p)\alpha^{\frac{n+2}{n-2}}\leq 0,
\end{equation}
which is impossible in the case $\alpha<1$. If $\alpha$ is not
attained in $\widetilde M$, then the same conclusion can be drawn from
Theorem \ref{max}. Thus $\inf _M \phi\geq 1$. A similar argument shows
that $\sup _M \phi \leq 1$, and $\phi _1=\phi _2$ follows.

\section{Acknowledgments}

The author thanks Lars Andersson, Piotr Chru\'sciel,
Mattias Dahl, Romain Gicquaud, Jim Isenberg, Lucy MacNay, and Hans Ringstr\"om for 
interesting and useful discussions during this work, and the Albert
Einstein Institute, Golm, for hospitality. 

The research was supported by the Knut and Alice 
Wallenberg Foundation and the Royal Swedish Academy
of Sciences.

\end{document}